\documentclass{amsart}
\usepackage{amsmath, amsthm, newlfont, amscd, amsfonts, amssymb, graphicx, color}
\usepackage[bookmarksnumbered, colorlinks, plainpages]{hyperref}
\newtheorem{thm}{Theorem}[section]
\newtheorem{cor}[thm]{Corollary}

\theoremstyle{definition}
\newtheorem{defn}[thm]{Definition}
\theoremstyle{remark}
\newtheorem{rem}[thm]{Remark}
\newtheorem{algo}[thm]{Algorithm}
\newtheorem{con}[thm]{Conjecture}
\newtheorem{ex}[thm]{Example}
\numberwithin{equation}{section}

\def\lead{\leaders\hbox to 1.5ex{\hss${.}$\hss}\hfill}
\def\arr{\hbox to 40pt{\rightarrowfill}}
\def\larr{\hbox to 40pt{\leftarrowfill}}

\long\def\alert#1{\parindent2em\smallskip\hbox to\hsize%
{\hskip\parindent\vrule%
\vbox{\advance\hsize-2\parindent\hrule\smallskip\parindent.4\parindent%
\narrower\noindent#1\smallskip\hrule}\vrule\hfill}\smallskip\parindent0pt}

\begin{document}

\title[Realization of Lie algebras of high dimension...]{Realization of Lie algebras of high dimension via pseudo-bosonic operators}

\author[F. Bagarello]{Fabio Bagarello}
\address{Dipartimento di  Ingegneria \endgraf
Universit\'a  di Palermo\endgraf
Viale delle Scienze, I-90128\endgraf
Palermo, Italy\endgraf
and \endgraf
Istituto Nazionale di Fisica Nucleare \endgraf
Sezione di Napoli, via Cinthia, Montesantangelo\endgraf
Napoli, Italy \endgraf 
URL: www1.unipa.it/fabio.bagarello \endgraf}
\email{fabio.bagarello@unipa.it}

\author[F.G. Russo]{Francesco G. Russo}
\address{Department of Mathematics and Applied Mathematics\endgraf
 University of Cape Town \endgraf
 Private Bag X1, Rondebosch 7701\endgraf 
 Cape Town, South Africa\endgraf}
\email{francescog.russo@yahoo.com}


\date{\today}


\begin{abstract}
 The present paper is the third contribution of a series of works, where we investigate pseudo--bosonic operators and their connections with finite dimensional Lie algebras.  We show that  all finite dimensional nilpotent Lie algebras (over the complex field) can be realized by central extensions of Lie algebras of pseudo-bosonic operators.  
This result is interesting, because it provides new examples of dynamical systems for nilpotent Lie algebras of any dimension. One could ask whether these operators are intrinsic with the notion of nilpotence or not, but this is false. In fact we exibit both a simple Lie algebra and a solvable nonnilpotent Lie algebra, which can  be realized in terms of pseudo-bosonic operators.
\end{abstract}

\date{\today}

\maketitle

\section{Introduction}

Heisenberg algebras are finite dimensional Lie algebras, which are very common in  quantum mechanics and theoretical physics. Their  generalization has interested the researches  of several authors, because they describe very well the symmetries of many dynamical systems, and so Heisenberg algebras represent a first important step of mathematical modelisation of the natural phenomena. 

It is  possible to generalize Heisenberg algebras in terms of the structure, since they appear as appropriate extensions of two abelian Lie algebras under prescribed conditions. These  are reflected by commuting relations on the generators and can be formalized via the notion of nilpotence.

Therefore it is not surprising that a large part of the   classification of finite dimensional nilpotent Lie algebras   originate from motivations both of mathematics and physics. The references \cite{anch, beckekolman, degraaf, gong, morozov, mubarak, patera, turk, umlauf} are fundamental for the classification of Lie algebras of  dimension  $\le 6$. Seeley \cite{seeley}  shows the presence of infinitely many nonisomorphic finite dimensional nilpotent Lie algebras when their dimension is $ \ge 7$, improving an earlier result  \cite{sund1}, where the same problem is noted for dimension $9$. This fact represents an obstacle for the classification via dimension, and a classical approach is to look for large enough nilpotent ideals (see \cite{snobl&wint}), when one wants to go ahead to classify finite dimensional nilpotent Lie algebras of higher dimensions. The reference \cite{snobl} describes this approach.

An interpretation has been given in \cite{bag6, bag7}, introducing a framework of functional analysis and topology for pseudo-bosonic operators. This theory was developed in connection with $PT$ and similar extensions of standard quantum mechanics, where self adjoint Hamiltonians are replaced by other operators satisfying suitable symmetry conditions. The reader can refer to \cite{bagbook, bag1, bag5, bag6, bag7,ben,mosta} for the details. 

Here we connect the theory of pseudo-bosonic operators and that of finite dimensional complex Lie algebras, showing that those which are nilpotent can always be realized by  systems of pseudo-bosons. Actually we give some examples, where the hypothesis of nilpotence can be omitted.

The organization of the present paper is the following. Section 2 describes a short survey on a method of construction of Skjelbred and Sund \cite{sund1}, which turns out to be a powerful tool for the proof of our main result. Additional feedback on the theory of the pseudo-bosons is recalled at the beginning of Section 3, before to prove our main theorems, which are Theorems \ref{main1}, \ref{main2}  and \ref{main3}. A final conjecture is open, since we haven't  found yet a finite dimensional complex nonnilpotent nonsolvable Lie algebra, which cannot be realized by pseudo-bosonic operators.

\section{A brief survey on the method of Skjelbred and Sund}

A complete classification of nilpotent Lie algebras has interested several authors since long time and the presence of nonisomorphic finite dimensional nilpotent Lie algebras forces the use of abstract methods along with those of computational nature. Umlauf \cite{umlauf} gave a first classification of nilpotent algebras up to dimension six, but it is well known that his results were affected by mistakes, because several algebras presented isomorphisms among them and this gave problems in the  identification.  After Umlauf's work  in 1891, there have been made several attempts to develop computational methods and alternative arguments of classification via dimension. Morozov \cite{morozov} observed
that there is a lower bound for the dimension of a maximal abelian ideal of a nilpotent Lie algebra. This suggests an inductive method, based on the idea to construct a nilpotent Lie algebra via extensions of others of low dimension.

Skjelbred and Sund \cite{sund1} have reduced the classification of nilpotent Lie algebras in a given dimension to the study of an appropriate space of representations  via the notion of Schur multiplier, which has been recalled below more formally.  Skjelbred and Sund's method is very interesting from a theoretical point of view, but it is not easy to use in practice, because the description of Schur multipliers has its own difficulties of computational nature. Neverthless these difficulties can be solved with the use of recent packages of computer algebra like \cite{gap}.

We are going to give more details on the method in \cite{sund1}, reporting also some observations which have been  made by Beck and Kolman \cite{beckekolman}. The classical notion of \textit{Chevalley-Eilenberg complex} for a Lie algebra can be found in \cite[Definition 7.7.1]{weibel} and that of Schur multiplier   in \cite[Chapter 7]{weibel}, but the main ideas are briefly sketched in our context of study. 

Consider a finite dimensional complex  Lie algebra $\mathfrak{l}$  with center \[Z(\mathfrak{l})=\{ a \in \mathfrak{l}  \ | \ [a,b]=0  \ \forall b \in \mathfrak{l}\} \] and derived Lie subalgebra \[[\mathfrak{l}, \mathfrak{l}]= \langle [a,b] \ | \ a,b \in \mathfrak{l}\rangle. \]
From now, until the end of the present paper, we will always refer to Lie algebras whose ground field is the field of complex numbers.

 The induced Lie algebra structure on  $Z(\mathfrak{l})$ reduces to that of a finite dimensional vector space with complex coefficients.

\medskip
\medskip
\medskip

\begin{defn}[See \cite{weibel}]\label{extensions}
We say that  $\mathfrak{l}$ is an \textit{extension} of an abelian Lie algebra $\mathfrak{a}$ (called \textit{abelian kernel}) by another  Lie algebra $\mathfrak{b}$, if there is a short exact sequence
$0 \ \longrightarrow \  \mathfrak{a} {\overset{\mu}{\longrightarrow}} \ \mathfrak{l} \ {\overset{\varepsilon}{\longrightarrow}} \mathfrak{b} \ \longrightarrow \ 0$
such that  $\mathfrak{a}$ is an abelian ideal of $\mathfrak{l}$ and $\mathfrak{l}/\mathfrak{a} \simeq \mathfrak{b}$ is the Lie algebra quotient. Here $\mu$ is  monomorphism,  $\varepsilon$  epimorphism, $\ker \varepsilon = \mathrm{Im} \ \mu $. If in addition $\mathfrak{a} \subseteq Z(\mathfrak{l})$,  then $\mathfrak{l}$ is called \textit{central extension} of $\mathfrak{a}$ by $\mathfrak{b}$.
\end{defn}

\medskip
\medskip
\medskip

A special case of extension, called  \textit{split extension}, is offered by the {\em semidirect sum} of two of Lie subalgebras $\mathfrak{a}$ and $\mathfrak{b}$ of a Lie algebra $\mathfrak{l}$. This happens when $\mathfrak{a}$  is an ideal of $\mathfrak{l}$, $\mathfrak{l}= \mathfrak{a} + \mathfrak{b}$ and  $\mathfrak{a} \cap \mathfrak{b} = 0$.
Of course, Definition \ref{extensions} is satisfied and it is instructive  to  produce examples of semidirect sums, which are just extensions but not necessarily central extensions.

\medskip
\medskip
\medskip

\begin{defn}\label{uppercentralseries}The \textit{upper central series} of $\mathfrak{l}$ is the series
\[0=Z_0(\mathfrak{l}) \leq Z_1(\mathfrak{l})\leq Z_2(\mathfrak{l}) \leq \ldots \leq Z_i(\mathfrak{l}) \le Z_{i+1}(\mathfrak{l}) \leq \ldots, \]
where each $Z_i(\mathfrak{l})$ turns out to be an ideal of $\mathfrak{l}$ (called the $i$-th $center$ of $\mathfrak{l}$) defined by \[ \frac{Z_1(\mathfrak{l})}{Z_0(\mathfrak{l})}=Z(\mathfrak{l}),\quad \frac{Z_2(\mathfrak{l})}{Z_1(\mathfrak{l})}=Z\left(\frac{\mathfrak{l}}{Z_1(\mathfrak{l})}\right),   \ldots, \frac{Z_{i+1}(\mathfrak{l})}{Z_i(\mathfrak{l})}=Z\left(\frac{\mathfrak{l}}{Z_i(\mathfrak{l})}\right), \ldots \]
We say that $\mathfrak{l}$ is  \textit{nilpotent of class $c$} if  the upper central series of $\mathfrak{l}$ ends after   $c $ of steps. By duality, one can introduce the \textit{lower central series} of $\mathfrak{l}$,  defined in terms of commutators by
\[\gamma_1(\mathfrak{l})=\mathfrak{l} \geq \gamma_2(\mathfrak{l})=[\mathfrak{l},\mathfrak{l}]  \geq  \gamma_3(\mathfrak{l})=[[\mathfrak{l},\mathfrak{l}],\mathfrak{l}]  \geq \ldots \geq \gamma_i(\mathfrak{l}) \geq \gamma_{i+1}(\mathfrak{l}) \geq \ldots, \]
where  $\gamma_i(\mathfrak{l})$ turns out to be an ideal of $\mathfrak{l}$ (called $i$th $derived$ of $\mathfrak{l}$) 
 and $\mathfrak{l}$ is  \textit{nilpotent of class $c$} if $\gamma_{c+1}(\mathfrak{l})=0$.
 
The \textit{derived series} of  $\mathfrak{l}$ is the series
\[\mathfrak{l} \supseteq [\mathfrak{l},\mathfrak{l}]=\mathfrak{l}^{(1)} \supseteq  [\mathfrak{l}^{(1)},\mathfrak{l}^{(1)}]=\mathfrak{l}^{(2)} \supseteq \ldots \supseteq \mathfrak{l}^{(i+1)}=[\mathfrak{l}^{(i)},\mathfrak{l}^{(i)}] \supseteq \ldots \]
 and one can check that each quotient $\mathfrak{l}^{(i+1)}/\mathfrak{l}^{(i)}$ is abelian Lie algebra. We say that $\mathfrak{l}$ is  \textit{solvable of derived length $m$} if  the derived series reaches the zero after   $m $  steps. 
In case $m \le 2$, $\mathfrak{l}$ is  called \textit{metabelian}\footnote{Definition \ref{uppercentralseries} is explicitly reported, because the adjective ``metabelian'' might have a different meaning in Russian literature.  From Definition \ref{uppercentralseries},  ``metabelian'' is `` of derived length $\le 2$'', so a centerless Lie algebra of derived length $2$ is metabelian and nonnilpotent, while ``metabelian'' might mean ``nilpotent of class $\le 2$'' according to some papers which we encountered in literature. This is the case of \cite[Definition 1.1]{gal}.}.  
 \end{defn}

\medskip
\medskip
\medskip

There exist finite dimensional solvable nonnilpotent Lie algebras in \cite{patera}.

\medskip
\medskip
\medskip

\begin{ex}\label{nonnilpotent}For all $\alpha \in \mathbb{C} \setminus \{0\}$, the Lie algebra
\[L_{3,5}= \langle x_1, x_2, x_3 \ | \ [x_1,x_2]=x_3, [x_1,x_3]=\alpha x_2, [x_2,x_1]=- x_3, [x_3,x_1]=\alpha x_2\rangle\]
 is metabelian. Note that $[x_1,x_2], [x_2,x_1], [x_1,x_3], [x_3,x_1] \in \langle x_2, x_3 \rangle = \langle x_2 \rangle \oplus \langle x_3 \rangle \simeq \mathfrak{i} \oplus \mathfrak{i}$ is an abelian Lie algebra of dimension two. Moreover $[\mathfrak{l},\mathfrak{l}]= \langle x_2,x_3\rangle$ is clearly an ideal of $\mathfrak{l}$,  $\langle x_2,x_3\rangle \cap \langle x_1\rangle=0$ and  $ \mathfrak{l}/[\mathfrak{l},\mathfrak{l}] = \langle x_1\rangle$ is one dimensional. The fact that it is metabelian and nonabelian is shown by the defining relations. In addition, the commuting relations show that there is no nontrivial center, i.e.: $Z(\mathfrak{l})=0$. Note also that $L_{3,5}$  is split extension of $\langle x_2,x_3\rangle$ by $\langle x_1 \rangle$; actually it is their semidirect sum. This shows an example of metabelian nonnilpotent Lie algebra. Further examples of solvable nonnilpotent Lie algebras of dimension $3$  can be found in  \cite{patera}.
\end{ex}

\medskip
\medskip
\medskip

Two observations are appropriate here. The first is that \cite{nd} (and more generally \cite{snobl}) provide examples of finite dimensional solvable Lie algebras, whose structure  generalizes that of  Example \ref{nonnilpotent}. This is done via the notion of \textit{nilradical} $N(\mathfrak{l})$ of $\mathfrak{l}$, that is, $N(\mathfrak{l})$ is the largest nilpotent ideal in $\mathfrak{l}$. For instance, $N(\mathfrak{l})=\langle x_2, x_3 \rangle$ in Example \ref{nonnilpotent} and one can easily see that, despite $\mathfrak{l}$ is nonnilpotent, it has a large nilpotent ideal $N(\mathfrak{l})$, which may be used to get information on $\mathfrak{l}$. This idea helps to classify nonnilpotent finite dimensional Lie algebras for high dimensions. A second important observation is in fact the physical meaning of such Lie algebras and this  has been investigated since long time, see \cite{snobl&wint}.
 
Of course, nilpotent Lie algebras are solvable and some important  examples are offered by the Heisenberg algebras which are reported below.

\medskip
\medskip
\medskip

\begin{defn}\label{Heisenberg} We say that  $\mathfrak{l}$ is  $Heisenberg$ provided that $[\mathfrak{l},\mathfrak{l}]=Z(\mathfrak{l})$ and $\mathrm{dim}([\mathfrak{l},\mathfrak{l}]) = 1$. Such algebras are odd dimensional with
basis $v_1, \ldots , v_{2m}, v$ and the only nonzero commutator between basis elements is $[v_{2i-1}, v_{2i}] = -
[v_{2i}, v_{2i-1}]= v$ for $i = 1,2, \ldots ,m$. They are  denoted by $\mathfrak{h}(m)$
and $\mathrm{dim} \ \mathfrak{h}(m)=2m + 1$.
\end{defn}

\medskip
\medskip
\medskip

At this point, we adapt some notions from \cite[Pages 642--644]{degraaf} to our context of study. Assume $\mathfrak{l}$ to be nilpotent  of finite dimension on $\mathbb{C}$, $V$ vector space on $\mathbb{C}$ and  recall that a skew--symmetric bilinear map on $\mathfrak{l}$ is a map $$\theta : (x , y) \in  \mathfrak{l} \times \mathfrak{l} \mapsto \theta(x , y)  \in V$$ satisfying for all $x,y,z \in \mathfrak{l}$ the  condition 
\begin{equation}\label{cycle}\theta([x,y],z) + \theta([z,x],y) + \theta([y,z],x)= 0.
\end{equation}
This turns out to be a well defined map, called  \textit{cocycle} (of length $2$) in $\mathfrak{l}$. The set  $Z^2(\mathfrak{l})$ of all cocycles on $\mathfrak{l}$ forms an abelian Lie algebra of finite dimension. Now  if we assign a linear map $\nu : u\in \mathfrak{l} \to \nu(u) \in V$ and define $$\beta :  (x, y) \in  \mathfrak{l} \times \mathfrak{l} \mapsto  \beta (x, y) = \nu([x, y]) \in V,$$ then we can check that  \eqref{cycle} is satisfied, so $\beta$ is a particular type of cocyle, called \textit{coboundary} (of length $2$) in $\mathfrak{l}$. The set $B^2(\mathfrak{l})$  of all coboundaries on $\mathfrak{l}$ forms a subalgebra of $Z^2(\mathfrak{l})$ and the \textit{Schur multiplier}  of $\mathfrak{l}$ is defined by
$$M(\mathfrak{l})= \frac{Z^2(\mathfrak{l})}{B^2(\mathfrak{l}) }.$$
According to \cite{weibel}, it describes the second cohomology Lie algebra with complex coefficients. Of course $M(\mathfrak{l})$ is an abelian Lie algebra.

Given $\theta \in Z^2(\mathfrak{l}) $, we may write its kernel and its radical  
\begin{equation}\label{radical}
 \ker \theta =\{(x,y)  \ |  \ \theta (x,y)=0 \}, \ \ \   \theta^+ =\{x\in \mathfrak{l}  \ |  \ \theta (x,y)=0,  \ \ \forall y \in \mathfrak{l} \},
\end{equation} 
respectively. We may also introduce a new set \begin{equation}\label{delicate}\mathfrak{l}_\theta = \mathfrak{l} \oplus V
\end{equation}
and endow it of the Lie algebra structure, defining for all $x, y  \in \mathfrak{l}$ and $u,v \in V$, that is, for all $(x+u, y+v) \in \mathfrak{l}_\theta \times \mathfrak{l}_\theta$ the Lie bracket 
\begin{equation}\label{liebracket}
[x+  u, \ y + v] = {[x,y]}_\mathfrak{l} + \theta(x,y)
\end{equation}   on $\mathfrak{l}_\theta,$ where ${[ \ , \ ]}_\mathfrak{l}$ denotes the Lie bracket operation in $\mathfrak{l}$. We are basically extending the Lie bracket in $\mathfrak{l}$ with an additive term in \eqref{liebracket}, depending only on $\theta$. One can easily check that this turns out to be a new Lie bracket operation, but this time in  $\mathfrak{l}_\theta.$ Of course, Definition \ref{extensions} applies to $\mathfrak{l}_\theta$ and in fact $\mathfrak{l}_\theta$ is a  new complex Lie algebra of finite dimension which appears as central extension of $\mathfrak{l}$.  It is also elementary to check that $\mathfrak{l}_\theta \simeq \mathfrak{l}_{\theta + \mu}$ for all $\mu \in B^2(\mathfrak{l})$, that is, the construction of the central extension is unique up to coboundaries. Therefore, one can assume $\theta \in M(\mathfrak{l})$ without loss of generality, in order to form central extensions.
The aforementioned process holds in particular when we replace the role of $\mathfrak{l}$ with $\mathfrak{l}/Z(\mathfrak{l})$ and that of $V$ with $Z(\mathfrak{l}) \neq 0$. In fact one can see  that any Lie algebra with a nontrivial centre can be obtained as a central extension of a Lie algebra of smaller dimension. So in particular, all nilpotent Lie algebras can be
constructed this way. Note also that when we construct nilpotent Lie algebras as $\mathfrak{l}_\theta$, we may restrict to $\theta$ such that $Z(\mathfrak{l}_\theta)=V$, avoiding to construct the same Lie algebra as central extension of different Lie algebras.
Details  can be found in \cite{degraaf}. We have all that we need, in order to recall the following result.

\medskip
\medskip

\medskip
\medskip

\begin{algo}[Skjelbred and Sund \cite{sund1}, 1978]\label{method}
Let $\mathfrak{l}$ be a finite dimensional nilpotent Lie algebra with center $Z(\mathfrak{l})= \langle u_1, u_2, \ldots, u_z \rangle$.
\begin{itemize}

\medskip
\medskip

\item[\textbf{Step 1.}] Compute $Z^2(\mathfrak{l}/Z(\mathfrak{l}))$, $B^2(\mathfrak{l}/Z(\mathfrak{l}))$ and so $M(\mathfrak{l}/Z(\mathfrak{l}))$.

\medskip
\medskip

\item[\textbf{Step 2.}] Note that two cocycles $\theta_1, \theta_2 \in Z^2(\mathfrak{l}/Z(\mathfrak{l}))$ induce the same central extension if they are equivalent modulo  $B^2(\mathfrak{l}/Z(\mathfrak{l}))$, so, it is possible to consider a unique  extension $\mathfrak{l}/Z(\mathfrak{l})_\theta$ with $\theta \in M(\mathfrak{l}/Z(\mathfrak{l}))$ up to isomorphisms.

\medskip
\medskip

\item[\textbf{Step 3.}] Give  $\theta \in M(\mathfrak{l}/Z(\mathfrak{l}))$ and write the linear combination $\theta(x+Z(\mathfrak{l}),y+Z(\mathfrak{l}))=\theta_1(x+Z(\mathfrak{l}),y+Z(\mathfrak{l}))u_1 + \ldots + \theta_z (x+Z(\mathfrak{l}),y+Z(\mathfrak{l}))u_z$
for linear independent cocycles $\theta_1, \theta_2, \ldots \theta_z \in M(\mathfrak{l}/Z(\mathfrak{l}))$ such that the condition $ \theta^+ \cap Z(\mathfrak{l}/Z(\mathfrak{l}))=0$ is satisfied.

\medskip
\medskip

\item[\textbf{Step 4.}] Note that an element $\varphi$ of the automorphism group $\mathrm{Aut}(\mathfrak{l}/Z(\mathfrak{l}))$ acts on a cocycle $\alpha \in Z^2(\mathfrak{l}/Z(\mathfrak{l}))$ via $\varphi(\alpha(x+Z(\mathfrak{l}),y+Z(\mathfrak{l})))= \alpha(\varphi(x+Z(\mathfrak{l}),\varphi(y+Z(\mathfrak{l}))))$. Consequently $\mathrm{Aut}(\mathfrak{l}/Z(\mathfrak{l}))$ acts on $M(\mathfrak{l}/Z(\mathfrak{l}))$  modulo $B^2(\mathfrak{l}/Z(\mathfrak{l}))$.

\medskip
\medskip

\item[\textbf{Step 5.}] Find the orbits of the action of $\mathrm{Aut}(\mathfrak{l}/Z(\mathfrak{l}))$  on the given  $ \theta \in M(\mathfrak{l}/Z(\mathfrak{l}))$ in Step 3 above.

\medskip
\medskip

\item[\textbf{Step 6.}] For each $\theta$, we get the Lie algebra $\mathfrak{l}/Z(\mathfrak{l})_\theta=\mathfrak{l}/Z(\mathfrak{l}) \oplus Z(\mathfrak{l})$.

\end{itemize}
\end{algo}

\medskip
\medskip

Because of Algorithm \ref{method}, it is possible to construct all nilpotent Lie algebras of finite dimension. Step 1 does not present difficulties for small dimensions, but higher dimensions  require the help of appropriate programs like  \cite{gap}. A similar comment must be made for Step 5.

\medskip
\medskip

\begin{rem}\label{new}
It is well known that a nilpotent finite dimensional Lie algebra can be always described via central extensions of smaller finite dimensional nilpotent Lie algebras (see \cite{seeley, weibel}) in an abstract way, but the concrete realisation  can give a series of complications of computational nature, since the knowledge of the Schur multipliers must be involved directly. Therefore we presented Algorithm  \ref{method}, which is significant and helps to visualise concrete steps of  construction of finite dimensional nilpotent Lie algebras via smaller ones.
\end{rem}

\medskip
\medskip

As application of this method, we show the following fact which is well known but reported just to convenience of the reader.

\begin{cor}\label{fundamental}
There are only two nilpotent Lie algebras of dimension $3$, namely the abelian Lie algebra $\mathfrak{i} \oplus \mathfrak{i} \oplus \mathfrak{i}$, where $\mathrm{dim} \ \mathfrak{i} =1$,   and $\mathfrak{h}(1)$. They induce all the nilpotent Lie algebras of dimension $4$, which are exactly $\mathfrak{i} \oplus \mathfrak{i} \oplus \mathfrak{i}$, $\mathfrak{h}(1) \oplus \mathfrak{i}$ and $\mathfrak{l}_{4,3} = \langle v_1, v_2, v_3, v_4 \ | \ [v_1,v_2]=v_3, [v_1,v_3]=v_4 \rangle$, via Algorithm \ref{method}.
\end{cor}

It is useful to compare Corollary \ref{fundamental} with  \cite[Lemma 2.10]{bag7}, where we find some five dimensional nilpotent Lie algebras which are constructed as central extensions of nilpotent Lie algebras of dimension four. In fact both Corollary \ref{fundamental} and \cite[Lemma 2.10]{bag7} adapt ideas and techniques of \cite{st3, st4}.

Note that Algorithm \ref{method} applies only to finite dimensional nilpotent Lie algebras, in fact classical results  \cite[Chapter 6]{hofmor} show that this is no longer possible when we remove the finite dimension or the nilpotence. We are going to recall some notions, which can be found in \cite{hofmor} with a topological emphasis and in  \cite{weibel} as well. From \cite[Definition 7.8.1]{weibel}, a (complex) Lie algebra $\mathfrak{l}$ is \textit{simple}, if it has no ideals except itself and $0$, and if $\mathfrak{l}=[\mathfrak{l},\mathfrak{l}]$. The largest solvable ideal in a Lie algebra $\mathfrak{l}$ is called \textit{solvable radical} and  denoted by $S(\mathfrak{l})$. Of course,  $N(\mathfrak{l}) \subseteq S(\mathfrak{l})$ but Example \ref{nonnilpotent} shows that the equality is false in general.

According to \cite[Definition 7.8.1]{weibel}, a \textit{semisimple Lie algebra} $\mathfrak{l}$ is a Lie algebra such that $S(\mathfrak{l})=0$, that is, a Lie algebra with no nonzero solvable ideals. The structure of semisimple Lie algebras  \cite[Theorem 7.8.5]{weibel} shows that a complex finite dimensional Lie algebra is semisimple if and only if  
it can be decomposed in direct sum of finitely many  (complex) simple Lie algebras. Moreover,

\medskip
\medskip

\begin{thm}[See Theorem 7.8.13, Levi Decomposition, in \cite{weibel}]
If $\mathfrak{l}$ is a finite dimensional complex Lie algebra, then $\mathfrak{l}$ is the semidirect sum of $\mathfrak{a}=S(\mathfrak{l})$ and $\mathfrak{b} \simeq \mathfrak{l}/S(\mathfrak{l})$ which is semisimple.
\end{thm}

The complex Lie algebra $$\mathfrak{sl}(2)=\{x \in \mathfrak{gl}(2)   \ | \ \mathrm{trace}(x)=0\}$$ is simple, according to  \cite[Definition 7.8.1]{weibel}. If $x^t$ is the transpose of $x \in \mathfrak{gl}(2)$, then $$\mathfrak{so}(3)=\{x \in \mathfrak{gl}(2)   \ | \ x+x^t=0\}$$
and  we may define a scalar product on $\mathfrak{sl}(2)$ by $\langle y, z \rangle= \mathrm{trace}(yz)$. The adjoint $$\mathrm{ad} \ : \ x \in \mathfrak{sl}(2) \mapsto \mathrm{ad}(x) \in \mathrm{Der}(\mathfrak{sl}(2))$$
is an action which preserves the scalar product, that is, for all $x,y,z \in \mathfrak{sl}(2)$. Of course,  $ \mathrm{Der}(\mathfrak{sl}(2))$ denotes the set of all derivations in  $\mathfrak{sl}(2)$ and we find that
$$\langle \mathrm{ad}(x) y, z \rangle + \langle y, \mathrm{ad}(x)z \rangle=0.$$
Therefore the adjoint action induces a Lie homomorphism from $\mathfrak{sl}(2)$ to $\mathfrak{so}(3)$, which actually is an  isomorphism such that \begin{equation}\label{technical}\mathfrak{sl}(2) \simeq \mathfrak{so}(3).
\end{equation}
This fact will be used in one of the main results of the next section. In particular, $\mathfrak{so}(3)$ is a simple Lie algebra of dimension 3.

\section{Main results}
In the present section we first list few facts on pseudo-bosons which are relevant for us. More results can be found in \cite{bag1,bag6,bag7}.

 Let $\mathcal{H}$ be a given Hilbert space (over the field $\mathbb{C}$ of complex numbers) with scalar product $\left<.,.\right>$ and related norm $\|.\|$.
Let $a$ and $b$ be two operators
on $\mathcal{H}$, with domains $D(a)$ and $D(b)$ respectively, $a^\dagger$ and $b^\dagger$ their adjoint, and let $\mathcal{D}$ be a dense subspace of $\mathcal{H}$
such that $a^\sharp\mathcal{D}\subseteq\mathcal{D}$ and $b^\sharp \mathcal{D} \subseteq \mathcal{D}$, where with $x^\sharp$ we indicate $x$ or $x^\dagger$. Of course, $\mathcal{D}\subseteq D(a^\sharp)$
and $\mathcal{D}\subseteq D(b^\sharp)$.

\medskip
\medskip

\begin{defn}\label{def21}
	The operators $(a,b)$ are $\mathcal{D}$-\textit{pseudo-bosonic}  if, for all $f\in\mathcal{D}$, we have
$	a\,b\,f-b\,a\,f=f.$
\end{defn}

\medskip
\medskip

The working assumptions are the following:

\vspace{2mm}

{\bf Assumption $\mathcal{D}$-pb 1.--}  there exists a nonzero $\varphi_{ 0}\in\mathcal{D}$ such that $a\,\varphi_{ 0}=0$.

\vspace{1mm}

{\bf Assumption $\mathcal{D}$-pb 2.--}  there exists a zero $\Psi_{ 0}\in\mathcal{D}$ such that $b^\dagger\,\Psi_{ 0}=0$.

\vspace{2mm}

It is obvious that, since $\mathcal{D}$ is stable under the action of the operators introduced above,  $\varphi_0\in D^\infty(b)=\cap_{k\geq0}D(b^k)$ and  $\Psi_0\in D^\infty(a^\dagger)$, so
that the vectors
\begin{equation}\label{A2}
\varphi_n=\frac{1}{\sqrt{n!}}\,b^n\varphi_0,\qquad \Psi_n=\frac{1}{\sqrt{n!}}\,{a^\dagger}^n\Psi_0,
\end{equation}
$n\geq0$, can be defined and they all belong to $\mathcal{D}$. Then, they also belong to the domains of $a^\sharp$, $b^\sharp$ and $N^\sharp$, where $N=ba$.

We see that, from a practical point of view, $\mathcal{D}$ is the natural space to work with and, in this sense, it is even more relevant than $\mathcal{H}$. Let's put $$\mathcal{F}_\Psi=\{\Psi_{ n} \ |  \ n\geq0\} \ \mathrm{and} \  \mathcal{F}_\varphi=\{\varphi_{ n} \ | \ n\geq0\}.$$
It is  simple to deduce the following lowering and raising relations:
\begin{equation}\label{A3}
\left\{
\begin{array}{ll}
b\,\varphi_n=\sqrt{n+1}\varphi_{n+1}, \qquad\qquad\quad\,\, n\geq 0,\\
a\,\varphi_0=0,\quad a\varphi_n=\sqrt{n}\,\varphi_{n-1}, \qquad\,\, n\geq 1,\\
a^\dagger\Psi_n=\sqrt{n+1}\Psi_{n+1}, \qquad\qquad\quad\, n\geq 0,\\
b^\dagger\Psi_0=0,\quad b^\dagger\Psi_n=\sqrt{n}\,\Psi_{n-1}, \qquad n\geq 1,\\
\end{array}
\right.
\end{equation}
as well as the eigenvalue equations $N\varphi_n=n\varphi_n$ and  $N^\dagger\Psi_n=n\Psi_n$, $n\geq0$.

 In particular, as a consequence
of these last two equations,  if we choose the normalization of $\varphi_0$ and $\Psi_0$ in such a way that $\left<\varphi_0,\Psi_0\right>=1$, then we may deduce
\begin{equation} \label{A4} \left<\varphi_n,\Psi_m\right>=\delta_{n,m},
\end{equation}
for all $n, m\geq0$. Hence $\mathcal{F}_\Psi$ and $\mathcal{F}_\varphi$ are biorthogonal, even if they are not necessarily bases for $\mathcal{H}$. This aspect, not particularly relevant here, is widely discussed in \cite{bag1}.

In \cite{bag6} we have discussed an algebraic point of view for the operators $a$, $b$ and their adjoints, which we briefly review here.

We have stressed in several occasions that  $a$, $b$ and their adjoints are not bounded operators, so that there is no  $C^*$-algebra $B(\mathcal{H})$ such that $a,b\in B(\mathcal{H})$. We refer to \cite{schu} for a monograph on unbounded operators theory.
However, as discussed in \cite{bag6}, we can still construct some algebraic settings for them. In particular,   {\em quasi $*$-algebras} and $O^*$-$algebras$ work well (see \cite{aitbook} for more details on these notions).

For instance we have:

\medskip
\medskip

\begin{defn}\label{o*}Let $\mathcal{H}$ be a separable Hilbert space and $N_0$ an
	unbounded, densely defined, self-adjoint operator. Let $D(N_0^k)$ be
	the domain of the operator $N_0^k$, $k \ge 0$, and $\mathcal{D}$ the domain of
	all the powers of $N_0$, that is,  $$ \mathcal{D} = D^\infty(N_0) = \bigcap_{k\geq 0}
	D(N_0^k). $$ This set is dense in $\mathcal{H}$. Let us now introduce
	$\mathcal{L}^\dagger(\mathcal{D})$, the $*$-algebra of all the \textit{  closable operators}
	defined on $\mathcal{D}$ which, together with their adjoints, map $\mathcal{D}$ into
	itself. Here the adjoint of $X\in\mathcal{L}^\dagger(\mathcal{D})$ is
	$X^\dagger=X^*_{| \mathcal{D}}$. Such an $\mathcal{L}^\dagger(\mathcal{D})$ is called  $O^*$-algebra.
\end{defn}

\medskip
\medskip

In $\mathcal{D}$ the topology is defined by the following $N_0$-depending
seminorms: $$\phi \in \mathcal{D} \rightarrow \|\phi\|_n\equiv \|N_0^n\phi\|,$$
where $n \ge 0$, and  the topology $\tau_0$ in $\mathcal{L}^\dagger(\mathcal{D})$ is introduced by the seminorms
$$ X\in \mathcal{L}^\dagger(\mathcal{D}) \rightarrow \|X\|^{f,k} \equiv
\max\left\{\|f(N_0)XN_0^k\|,\|N_0^kXf(N_0)\|\right\},$$ where
$k \ge 0$ and   $f \in \mathcal{C}$, the set of all the positive,
bounded and continuous functions  on $\mathbb{R}_+$, which are
decreasing faster than any inverse power of $x$:
$\mathcal{L}^\dagger(\mathcal{D})$ is a {   complete *-algebra} with respect to the topology $\tau_0$.

As a consequence, if $x,y\in \mathcal{L}^\dagger(\mathcal{D})$, we can multiply the two elements and the results, $xy$ and $yx$, both belong to $\mathcal{L}^\dagger(\mathcal{D})$. This is what we were looking for: a suitable structure in which we have the possibility of introducing Lie brackets for objects (linear operators, in our case), which are not everywhere defined.
In fact, for each pair $x,y\in\mathcal{L}^\dagger(\mathcal{D})$, we can define a map $[.,.]$ as follows:
\begin{equation}\label{B1}
[x,y]=xy-yx.
\end{equation}

It is clear that $[x,y]\in\mathcal{L}^\dagger(\mathcal{D})$, that it is bilinear, that $[x,x]=0$ for all $x\in\mathcal{L}^\dagger(\mathcal{D})$, and that it satisfies the Jacobi identity
$$
[x,[y,z]]+[y,[z,x]]+[z,[x,y]]=0,
$$
for all $x,y,z\in\mathcal{L}^\dagger(\mathcal{D})$. Therefore, $[.,.]$ is a Lie bracket defined on $\mathcal{L}^\dagger(\mathcal{D})$.

The operators $a$ and $b$, together with their adjoints, belong to the algebra of operators $\mathcal{L}^\dagger(\mathcal{D})$ for a suitable $\mathcal{D}$, see \cite{bag5,bag6}. With this in mind we can prove the following result.

\medskip
\medskip

\begin{thm}\label{main1}All finite dimensional complex nilpotent Lie algebras may be realized by central extensions of  pseudo-bosonic operators.
\end{thm}

\begin{proof} The main idea is to generalise Corollary \ref{fundamental} to arbitrary dimension. Assume $\mathfrak{l}$ is a nilpotent Lie algebra of $\mathrm{dim} \ \mathfrak{l} =n$ and  work by induction on $n$. Since there are no nilpotent Lie algebras when $n=2$, we must assume $n \ge 3$.

For $n=3$ the basis of induction is given by Corollary \ref{fundamental}. In fact we have seen that all finite dimensional nilpotent Lie algebras of dimension $4$ can be obtained by Algorithm \ref{method}, beginning from  $\mathfrak{i} \oplus \mathfrak{i}\oplus \mathfrak{i}$ and $\mathfrak{h}(1)$ which admit realizations by pseudo-bosons (see \cite{bag6, bag7}).
Then the result is true for $n=3$.

Assume the result is true for all nilpotent Lie algebras of dimension $<n$ when $n \ge 4$ and that $\mathrm{dim} \ Z(\mathfrak{l})=z \ge 1$. In particular, the result is true for $\mathfrak{l}/Z(\mathfrak{l})$, since  $$\mathrm{dim} \ \mathfrak{l}/Z(\mathfrak{l})= \mathrm{dim} \ \mathfrak{l} - \mathrm{dim} \ Z(\mathfrak{l})=n-z \le n-1 <n.$$ This means that we can realize  $\mathfrak{l}/Z(\mathfrak{l})$ as central extension of pseudo-bosonic operators. Now we apply all the steps of  Algorithm \ref{method} and get the central extension ${\mathfrak{l}/Z(\mathfrak{l})}_\theta$ for a chosen $\theta \in M(\mathfrak{l}/Z(\mathfrak{l}))$, where the Lie bracket of ${\mathfrak{l}/Z(\mathfrak{l})}_\theta$ has been defined  in \eqref{liebracket}. This ${\mathfrak{l}/Z(\mathfrak{l})}_\theta$ is the arbitrary finite dimensional nilpotent Lie algebra of dimension $n$, provided by Algorithm \ref{method}. Since  ${\mathfrak{l}/Z(\mathfrak{l})}_\theta$ may be written as \eqref{delicate}, it is the direct sum of an abelian Lie algebra by a nilpotent Lie algebra, realized by pseudo-bosonic operators (from induction hypothesis), where we introduce   the Lie bracket  ${[ \ , \ ]}_{{\mathfrak{l}/Z(\mathfrak{l})}_\theta}$ extending  ${[ \ , \ ]}_{\mathfrak{l}/Z(\mathfrak{l})}$ with the linear term $\theta(x+Z(\mathfrak{l}), y+Z(\mathfrak{l}))$ for all $x+Z(\mathfrak{l}),y +Z(\mathfrak{l}) \in \mathfrak{l}/Z(\mathfrak{l})$  (see \eqref{liebracket}). This means that the new commuting relations in ${\mathfrak{l}/Z(\mathfrak{l})}_\theta$ are the same we had in $\mathfrak{l}/Z(\mathfrak{l})$  with the addition of the commuting relations between $\mathfrak{l}/Z(\mathfrak{l})$ and $Z(\mathfrak{l})$ up to the action of $\theta$.

 Let's see this fact more formally: we have $$\mathfrak{l}/Z(\mathfrak{l})=\langle x_1 + Z(\mathfrak{l}), x_2 +Z(\mathfrak{l}), \ldots, x_{n-1} +Z(\mathfrak{l}) \ | \ R_1( x_1  +Z(\mathfrak{l}), x_2 +Z(\mathfrak{l}), \ldots, x_{n-1} +Z(\mathfrak{l})),$$
$$ \ldots,  R_k(x_1  +Z(\mathfrak{l}), x_2 +Z(\mathfrak{l}), \ldots, x_{n-1} +Z(\mathfrak{l})) \rangle,$$
which is an ($n-1$)-dimensional nilpotent Lie algebra with generators $x_1 + Z(\mathfrak{l}), x_2 +Z(\mathfrak{l}), \ldots, x_{n-1} +Z(\mathfrak{l})$ and $k$ relations $R_1( x_1  +Z(\mathfrak{l}), x_2 +Z(\mathfrak{l}), \ldots, x_{n-1} +Z(\mathfrak{l})),$
$ \dots, R_k( x_1  +Z(\mathfrak{l}), x_2 +Z(\mathfrak{l}), \ldots, x_{n-1} +Z(\mathfrak{l}))$ which can be expressed in terms of pseudobosons, that is, these relations are multilinear equations involving  the Lie brackets ${[ \ , \ ]}_{\mathfrak{l}/Z(\mathfrak{l})}$ and the generators of $\mathfrak{l}/Z(\mathfrak{l})$. Then we consider the abelian Lie algebra $Z(\mathfrak{l})=\langle x_n  \rangle$ of dimension one.

Forming the direct sum of $\mathfrak{l}/Z(\mathfrak{l})$ and $Z(\mathfrak{l})$, we get all the generators in $\mathfrak{l}/Z(\mathfrak{l})$ plus the generator of $Z(\mathfrak{l})$ and at the level of relations we have all the relations in $\mathfrak{l}/Z(\mathfrak{l})$ plus  the relations  $A(x_1  +Z(\mathfrak{l}), x_2 +Z(\mathfrak{l}), \ldots, x_{n-1} +Z(\mathfrak{l}); x_n)$ which come from the fact that the elements of $\mathfrak{l}/Z(\mathfrak{l})$ must commute with those of $Z(\mathfrak{l})$. Note that the presence of $A(x_1  +Z(\mathfrak{l}), x_2 +Z(\mathfrak{l}), \ldots, x_{n-1} +Z(\mathfrak{l}); x_n)$  is essential  by the definition of direct sum $\mathfrak{l}/Z(\mathfrak{l}) \oplus Z(\mathfrak{l})$; moreover $A(x_1  +Z(\mathfrak{l}), x_2 +Z(\mathfrak{l}), \ldots, x_{n-1} +Z(\mathfrak{l}); x_n)$ is clearly pseudo-bosonic. Until now we are considering the Lie algebra  
\begin{equation}\label{construction}
\langle x_1 + Z(\mathfrak{l}), x_2 +Z(\mathfrak{l}), \ldots, x_{n-1} +Z(\mathfrak{l}); x_n \ | \ 
\end{equation}
$$R_1( x_1  +Z(\mathfrak{l}), x_2 +Z(\mathfrak{l}), \ldots, x_{n-1} +Z(\mathfrak{l})),
 \ldots,  R_k(x_1  +Z(\mathfrak{l}), x_2 +Z(\mathfrak{l}), \ldots, x_{n-1} +Z(\mathfrak{l})),$$
 $$ A(x_1  +Z(\mathfrak{l}), x_2 +Z(\mathfrak{l}), \ldots, x_{n-1} +Z(\mathfrak{l}); x_n) \rangle$$
and this is $n$-dimensional  nilpotent and realized with pseudo-bosonic operators and pseudo-bosonic relations.

At this point we consider \eqref{liebracket} and observe that in case $\theta$ is trivial, then \eqref{construction} is exactly the Lie algebra ${\mathfrak{l}/Z(\mathfrak{l})}_\theta$ from Algorithm \ref{method} and  the result follows. In case   $\theta$ is nontrivial, then 
we note that endowing ${\mathfrak{l}/Z(\mathfrak{l})}_\theta$  of the Lie bracket \eqref{construction} is equivalent to endow ${\mathfrak{l}/Z(\mathfrak{l})}_\theta$ of the Lie bracket  on $\mathfrak{l}/Z(\mathfrak{l})$ modulo $\theta$, that is,  
\begin{equation}\label{modulotheta}
{[x+Z(\mathfrak{l}),y+Z(\mathfrak{l})]}_{\mathfrak{l}/Z(\mathfrak{l})}={[x+Z(\mathfrak{l}) +  u, \ y +Z(\mathfrak{l}) + v]}_{{\mathfrak{l}/Z(\mathfrak{l})}_\theta} -  \theta(x+Z(\mathfrak{l}),y +Z(\mathfrak{l})),  
\end{equation}
where $u, v \in Z(\mathfrak{l})$. This means that the pseudo-bosonic relations $$R_1( x_1  +Z(\mathfrak{l}), x_2 +Z(\mathfrak{l}), \ldots, x_{n-1} +Z(\mathfrak{l})),$$
$$ \ldots,  R_k(x_1  +Z(\mathfrak{l}), x_2 +Z(\mathfrak{l}), \ldots, x_{n-1} +Z(\mathfrak{l}))$$
can be written as multilinear equations of the form
\begin{equation}\label{modulothetabis}S_1( (x_1  +Z(\mathfrak{l})) + u_1, (x_2 +Z(\mathfrak{l})) + u_2, \ldots, (x_{n-1} +Z(\mathfrak{l})) + u_{n-1} ; \theta),
\end{equation}
$$ \ldots,  S_k((x_1  +Z(\mathfrak{l})) + u_1, (x_2 +Z(\mathfrak{l})) + u_2, \ldots, (x_{n-1} +Z(\mathfrak{l}))) + u_{n-1} ; \theta)$$
for all $u_1, u_2, \dots, u_{n-1} \in Z(\mathfrak{l}) $, and this means that we can write \eqref{modulothetabis} as pseudo-bosonic relations. Therefore the Lie algebra
\begin{equation}\label{constructionbis}
{\mathfrak{l}/Z(\mathfrak{l})}_\theta=\langle x_1 + Z(\mathfrak{l}), x_2 +Z(\mathfrak{l}), \ldots, x_{n-1} +Z(\mathfrak{l}); x_n \ | \ 
\end{equation}
$$S_1( (x_1  +Z(\mathfrak{l}))+ u_1, (x_2 +Z(\mathfrak{l}))+ u_2, \ldots, (x_{n-1} +Z(\mathfrak{l}))+ u_{n-1}; \theta),$$
 $$\ldots,  S_k((x_1  +Z(\mathfrak{l}))+ u_1, (x_2 +Z(\mathfrak{l}))+ u_2, \ldots, (x_{n-1} +Z(\mathfrak{l}))+ u_{n-1}; \theta),$$
 $$ A(x_1  +Z(\mathfrak{l}), x_2 +Z(\mathfrak{l}), \ldots, x_{n-1} +Z(\mathfrak{l}); x_n), \ \ \forall u_1, u_2, \ldots , u_{n-1} \in Z(\mathfrak{l}) \rangle$$
becomes the outcome of Algorithm \ref{method}: it is finite dimensional nilpotent and realized by pseudo-bosonic operators. The result follows.
 \end{proof}

\medskip
\medskip

As mentioned in \cite[Theorem 5.1]{bag6}, a realization of a dynamical system, related to the Lie algebras in Theorem \ref{main1}, is given, for instance, by the five dimensional nilpotent Lie algebra of the shifted harmonic oscillator.
\begin{equation}\label{ash1}
\mathfrak{a}_{sh} = \langle v_1, v_2, v_3, v_4, v \  |  \  [v_1,v_2]=[v_3,v_4]=[v_1,v_4]=[v_3,v_2]=v,
\end{equation}
where $v_1=a$, $v_2=b$, $v_3=b^\dagger$, $v_4=a^\dagger$ and $v=\mathbb{I}$, with 
\begin{equation}\label{ash2}
a=c-\alpha\mathbb{I},\qquad b=c^\dagger-\overline{\beta}\mathbb{I},
\end{equation}
with $\alpha, \beta \in \mathbb{C}$, $\alpha\neq\beta$, and where $c$ and $c^\dagger$ are the ordinary bosonic operators, satisfying $[c,c^\dagger]=\mathbb{I}$. Of course, $a$ and $b$ belong to $\mathcal{L}^\dagger(\mathcal{D})$, and similarly this happens for $a^\dagger$ and $b^\dagger$. With our choice,  $v_1$, $v_2$, $v_3$, $v_4$ and $v$ realize \eqref{ash1}. In fact it has been shown (see \cite[Theorem 6.1]{bag6}) that 
$$\mathfrak{a}_{sh} = \mathfrak{h}(1) \oplus ( \mathfrak{i} \oplus \mathfrak{i}),$$
where $\mathfrak{i}$ is abelian of dimension one.

Referring to  \cite{new}, we show that pseudo-bosonic operators may realize different finite dimensional Lie algebras which are centerless, so nilpotent.

\medskip
\medskip

\begin{thm}\label{main2}There exists a three dimensional simple Lie algebra, which can be realized by  pseudo-bosonic operators neither as central extension nor as semidirect sum.
\end{thm}

\begin{proof} 
Consider the pseudo-bosonic operators $A,B\in\mathcal{L}^\dagger(\mathcal{D})$ such that $[A,B]= \mathbb{I}$ and define 
$$L_0=\frac{1}{2}(BA+\frac{1}{2}\mathbb{I}), \ \ L_-=\frac{1}{2}A^2, \ \ L_+=\frac{1}{2}B^2.$$
Of course, they also belong to $\mathcal{L}^\dagger(\mathcal{D})$.

The following commuting relations are satisfied:
$$[L_0,L_-]=\frac{1}{4}[BA,A^2]=\frac{1}{4}[B,A^2]A=\frac{1}{4}(-2A)A=-\frac{1}{2}A^2=-L_-;$$
$$[L_0,L_+]=\frac{1}{4}[BA,B^2]=\frac{1}{4}B[A,B^2]=\frac{1}{4} \  B \ (2B)=\frac{1}{2}B^2=L_+;$$
$$[L_+,L_-]=\frac{1}{4}[B^2,A^2]=\frac{1}{4} \ (B[A,A^2] + [B,A^2]B)=\frac{1}{4} \ (B(-2A)-2AB)$$
$$= -\frac{1}{2} \ (AB+BA) =  -\frac{1}{2} \ (\mathbb{I}+BA+BA)=-(BA+\frac{1}{2}\mathbb{I})=-L_0.$$
Now the Lie algebra
$$\mathfrak{l}=\langle L_0, L_+, L_-  \ | \ [L_0,L_-]= -L_-, [L_0,L_+]=L_+, [L_+,L_-]=-L_0\rangle$$
has dimension three and is nilpotent because the commuting relations show that $Z(\mathfrak{l})=0$. The same commuting relations show that $\mathfrak{l}=[\mathfrak{l},\mathfrak{l}]$, that is, $\mathfrak{l}$ is a perfect Lie algebra (see also \eqref{technical}) and one can check without difficulties that $\mathfrak{l} \simeq \mathfrak{so}(3).$ 
\end{proof}

\medskip
\medskip

In \cite{new} similar operators were used to construct a generalized time dependent non-hermitian Swanson Hamiltonian, relevant in connection with $PT$-quantum mechanics, \cite{ben}. Our operators $A$ and $B$ above, and their combinations $L_0$, $L_-$, $L_+$, can be used to generalize further the Swanson Hamiltonian., and they can also used to define the so-called bi-squeezed states, relevant in quantum optics, \cite{bisqueezed}.

Note that Theorem \ref{main2} describes a simple Lie algebra, which of course is nonsolvable nonnilpotent. One could wonder whether there are examples of solvable nonnilpotent nonsimple Lie algebras which we may realize via pseudobosons. The answer is positive. 

\medskip
\medskip

\begin{thm}\label{main3}There exists a three dimensional  metabelian nonnilpotent  Lie algebra, which can be realized by pseudo-bosonic operators.
\end{thm}

\begin{proof} We consider a variation of Example \ref{nonnilpotent} and define the Lie algebra
\[\mathfrak{l}=\langle a_1, a_2, a_3 \ | \ [a_1, a_2]=a_2, \ [a_1,a_3]=2a_3 \rangle,\]
having in mind that the missing commutator relations give zero. Consider the pseudo-bosonic operators $A,B$ such that $[A,B]= \mathbb{I}$ and define 
$$a_1=BA, \ \ a_2=B, \ \ a_3=B^2.$$
Then  we get $[B,B^2]=[B^2,B]=0$,
$$[a_1,a_2] = [BA,B] = B [A,B] = B = a_2;$$
$$[a_1,a_3] = [BA,B^2] = B [A,B^2] = B ([A,B]B+B[A,B])= B \ (2B)= 2B^2=2a_3.$$
These relations show that $Z(\mathfrak{l})=0$, so $\mathfrak{l}$ is nonnilpotent. On the other hand, 
it is metabelian, because
$$[\mathfrak{l},\mathfrak{l}]=\langle [a_1,a_2], [a_2,a_1] , [a_1,a_3], [a_3,a_1] ,[a_2,a_3], [a_3,a_2] \rangle= \langle a_2, a_3\rangle = \langle a_2 \rangle \oplus \langle a_3 \rangle.$$
is two dimensional abelian and one can check that $\mathfrak{l}/[\mathfrak{l},\mathfrak{l}]   \simeq \langle a_1 \rangle$
is one dimensional in such a way that $\mathfrak{l}$ is semidirect sum of $[\mathfrak{l},\mathfrak{l}]$ by $ \langle a_1 \rangle $.
\end{proof}

\medskip
\medskip

A concrete situation can be presented, in order to illustrate Theorem \ref{main3}.

\medskip
\medskip

\begin{ex} The operators $a_1$, $a_2$ and $a_3$ of Theorem \ref{main3} can be used to construct a non-selfadjoint operator $H$,  which defines the dynamics of a one-mode quantum system, see \cite{baggarg} for instance. The Hamiltonian is
	$$
	H=\omega a_1+\lambda a_2+\mu a_3=\omega BA+\lambda B+\mu B^2,
	$$
	where $\omega, \lambda, \mu$ can be real or complex. This operator has a kinetic term $\omega a_1$ and two interaction contributions describing the creation of a single particle, $\lambda B$, or of a couple of particles simultaneously, $\mu B^2$.
\end{ex}

\medskip
\medskip

Due to the fact that metabelian finite dimensional Lie algebras can be constructed in an appropriate way via extensions of metabelian Lie algebras, we believe that, if there are examples of finite dimensional  Lie algebras, not realizable by pseudobosons, then these must be simple Lie algebras.

\medskip
\medskip

\begin{con}If there exists an example of a finite dimensional Lie algebra $\mathfrak{l}$, which cannot be realized by pseudo-bosonic operators, then $\mathfrak{l}$ must be a semisimple Lie algebra neither nilpotent nor solvable.
\end{con}

\end{document}